\documentclass[longbibliography, aps, prx, amsmath, amssymb, amsfonts, twocolumn, superscriptaddress, footinbib, floatfix, 10pt]{revtex4-2}

\usepackage[german,american]{babel}
\usepackage{graphicx} % Include figure files
\usepackage{dcolumn} % Align table columns on decimal point
\usepackage{bm} % bold math
\usepackage{amssymb}
\usepackage{amsmath}
\usepackage{amsfonts}
\usepackage{epsfig}
\usepackage{latexsym}
\usepackage{dcolumn}

\usepackage[ruled, vlined, linesnumbered]{algorithm2e}
\usepackage{hyperref}
\usepackage{float}
\usepackage[normalem]{ulem}
\usepackage{xcolor}
\usepackage{epstopdf}
\usepackage{physics}
\usepackage{lipsum}
\usepackage{microtype}

\usepackage{enumitem}

\usepackage{amsthm}

\newtheorem{theorem}{Theorem}
\newtheorem{corollary}{Corollary}

\newtheorem{definition}{Definition}

\hypersetup{                    
	colorlinks,
	citecolor=blue,
	filecolor=blue,
	linkcolor=blue,
	urlcolor=blue
}

\usepackage[utf8]{inputenc}
\usepackage{natbib}

\UseRawInputEncoding

\usepackage[margin=45pt]{geometry} % Adjusts all margins
\begin{document}
	
	\title{Investigating Pure State Uniqueness in Tomography via Optimization}

    \author{Jiahui Wu}
	
	\affiliation{Department of Physics, The Hong Kong University of Science and Technology, Clear Water Bay, Kowloon, Hong Kong, China}

	\author{Zheng An}
	
	\affiliation{Department of Physics, The Hong Kong University of Science and Technology, Clear Water Bay, Kowloon, Hong Kong, China}

	\author{Chao Zhang}
	
	\affiliation{Department of Physics, The Hong Kong University of Science and Technology, Clear Water Bay, Kowloon, Hong Kong, China}

    \author{Xuanran Zhu}
	
	\affiliation{Department of Physics, The Hong Kong University of Science and Technology, Clear Water Bay, Kowloon, Hong Kong, China}

    \author{Shilin Huang}

    %\email{huangsl@ust.hk}
	
	\affiliation{Department of Physics, The Hong Kong University of Science and Technology, Clear Water Bay, Kowloon, Hong Kong, China}

    \author{Bei Zeng}

    \email{bei.zeng@utdallas.edu}
	
	\affiliation{Department of Physics, The University of Texas at Dallas, Richardson, Texas 75080, USA}

	\date{\today}

    \begin{abstract}

    Quantum state tomography (QST) is crucial for understanding and characterizing quantum systems through measurement data. Traditional QST methods face scalability challenges, requiring $\mathcal{O}(d^2)$ measurements for a general $d$-dimensional state. This complexity can be substantially reduced to $\mathcal{O}(d)$ in pure state tomography, indicating that full measurements are unnecessary for pure states. In this paper, we investigate the conditions under which a given pure state can be uniquely determined by a subset of full measurements, focusing on the concepts of uniquely determined among pure states (UDP) and uniquely determined among all states (UDA). The UDP determination inherently involves non-convexity challenges, while the UDA determination, though convex, becomes computationally intensive for high-dimensional systems. To address these issues, we develop a unified framework based on the Augmented Lagrangian Method (ALM). Specifically, our theorem on the existence of low-rank solutions in QST allows us to reformulate the UDA problem with low-rank constraints, thereby reducing the number of variables involved. Our approach entails parameterizing quantum states and employing ALM to handle the constrained non-convex optimization tasks associated with UDP and low-rank UDA determinations. Numerical experiments conducted on qutrit systems and four-qubit symmetric states not only validate theoretical findings but also reveal the complete distribution of quantum states across three uniqueness categories: (A) UDA, (B) UDP but not UDA, and (C) neither UDP nor UDA. This work provides a practical approach for determining state uniqueness,
    advancing our understanding of quantum state reconstruction.
	\end{abstract}
	
	\maketitle
	
    \section{Introduction}
    
    Quantum state tomography (QST) is an process that involves reconstructing quantum states from measurement data~\cite{RevModPhys.81.299, PhysRevLett.109.120403, Qi2013, Paris2004}. This technique plays a fundamental role in understanding and verifying quantum systems, particularly within the rapidly evolving field of quantum technology~\cite{Haeffner2005, Kosaka2009, PhysRevLett.108.040502, Leibfried2005}. Traditional QST faces significant scalability challenge: in a $d$-dimensional Hilbert space, a general quantum state requires $d^2 - 1$ real parameters for complete description, necessitating $\mathcal{O}(d^2)$ independent measurements. This challenge has motivated extensive research into more efficient methods in recent years~\cite{PhysRevLett.105.150401, PhysRevLett.116.230501, Vanner2013, Carrasquilla2019, Ahmed2021, PhysRevApplied.21.014037}. A significant reduction in complexity emerges when the state to be reconstructed is known to be pure. A pure state $|\psi\rangle$ in $\mathbb{C}^d$ requires only $\mathcal{O}(d)$ measurements for determination~\cite{Heinosaari2013, chen2013uniqueness, ma2016pure}. This reduction makes pure state tomography especially valuable for high-dimensional quantum systems, where it can effectively decrease both measurement requirements and computational overhead during reconstruction.
    
    One promising approach to reconstruct pure quantum states is tomography via reduced density matrices (RDMs)~\cite{PhysRevLett.89.207901, PhysRevA.70.010302, PhysRevLett.89.277906, Chen2012, PhysRevA.86.022339, PhysRevA.88.012109}, which can significantly reduce experimental complexity while maintaining compatibility with practical quantum devices~\cite{PhysRevLett.118.020401}. Instead of performing measurements on the entire system, this method focuses on local measurements of subsystems, which constitute a subset of full measurements. The observation that a pure state may be reconstructed without requiring full measurements raises the following question: \newline 
    
    \textit{Under what conditions a given pure state can be uniquely specified by a subset of full measurements?}
    
    In this work, we focus on two concepts that closely related to this question: uniquely determined over all pure states (UDP) and uniquely determined over all states (UDA)~\cite{chen2013uniqueness}. Specifically, a pure state is considered UDP under a given measurement framework if no other pure state can produce identical measurement outcomes. A pure state is UDA if no quantum state, whether pure or mixed, can reproduce the same measurement results. While UDP ensures uniqueness within the pure state manifold, UDA extends this guarantee to all quantum states, accommodating scenarios involving imperfect state preparations or environmental interactions.
    
    The investigation of pure state tomography requires both rigorous theoretical analysis~\cite{Wang2020Pure, Huang2018, Sawicki2013} and practical numerical methods~\cite{Moroder2012,Jia2019, Xin2019,Zhang2023vjz}. Previous theoretical results suggest classifying many-body quantum states into three distinct and nontrivial  categories: (A) UDA, (B) UDP but not UDA, and (C) neither UDP nor UDA~\cite{PhysRevLett.118.020401}. Althought both UDP and UDA determination problems can be formulated as optimization problem, numerical validation of this classification has remained challenging. 
    This difficulty stems from two key factors: (i) the inherent non-convexity of UDP determination and (ii) the computational complexity of analyzing high-dimensional quantum systems.

    The UDP problem is inherently non-convex, due to the geometric structure of the pure state manifold. To address this constrained non-convex optimization challenge, we employ the Augmented Lagrangian Method (ALM). ALM, also known as the Method of Multipliers, was introduced by Hestenes~\cite{hestenes1969multiplier} and Powell~\cite{powell1969method} in 1969. This method augments the ordinary Lagrangian function with penalty terms to handle equality or inequality constraints more effectively~\cite{nocedal1999numerical,majumdar2020recent}.
    
    The UDA problem can be addressed with convex optimization techniques~\cite{baldwin2016strictly}, such as semidefinite programming (SDP), which guarantees global optimal solutions. However, when UDA problem involves high-dimensional variables, numerical computations become resource-intensive, limiting practical applicability. To alleviate this complexity, we prove the existence of low-rank solutions in QST under a subset of full measurements. This theorem enables us to reformulate the UDA problem with rank constraints, reducing the number of variables while introducing non-convexity. Notably, ALM can be also utilized to obtain low-rank feasible solutions to SDPs~\cite{burer2003nonlinear, wang2023solving}, suggesting it suitable for our low-rank UDA formulation. 
    
    Our practical approach to UDP problems, coupled with the low-rank formulation for UDA determinations, enables a comprehensive investigation of quantum state uniqueness in tomography. In this paper, we propose a unified framework based on ALM to address both UDP and low-rank UDA determinations. We conduct extensive numerical experiments on qutrit systems and four-qubit symmetric states, visualizing the distribution of these states across three distinct categories based on their uniqueness properties. Through these numerical experiments, we validate and extend previous theoretical results, providing deeper insights into quantum state reconstruction.
    
    The structure of this paper is organized as follows. In Section~\ref{sec:theory}, we provide the theoretical foundations for our work, introducing the concepts of UDP and UDA, and formulating the uniqueness problems as optimization problems. We then present our theorem on the existence of low-rank solutions in QST and discuss its implications. In Section~\ref{sec:method}, we describe our methodology, detailing the parameterization of quantum states and our implementation of ALM for determining uniqueness. Section~\ref{sec:numer} presents our numerical experiments and results for qutrit systems and four-qubit symmetric states, validating the theoretical results. Finally, in Section~\ref{sec:conclu}, we conclude with a discussion of our findings and outline potential directions for future research.
    
    \section{Theory}
    \label{sec:theory}
    
    \subsection*{Preliminaries}
    
    Quantum state tomography relies on two fundamental elements: the quantum states and the measurements applied to them. In a Hilbert space of dimension $d$, denoted as $\mathcal{H}^d$, a quantum state can be represented by a density operator $\rho$. This operator must satisfy two essential properties: it is positive semi-definite, and its trace equals one, ensuring normalization. Measurements are described by a set of observables $\mathbf{A} = \{A_1, A_2, \dots, A_{m}\}$.
    
    \begin{definition}[Measurement Vector]
    Given a measurement framework $\mathbf{A}$, measuring a quantum state $\rho$ produces a real-valued vector, known as the measurement vector:
    \begin{equation}
        \Vec{\mathcal{M}}_{\mathbf{A}}(\rho) = \begin{pmatrix}
            \Tr (A_1 \rho) \\
            \Tr (A_2 \rho) \\
            \vdots \\
            \Tr (A_m \rho) \\
        \end{pmatrix}.
    \end{equation}
    \end{definition}
    
    Building on these concepts, we define the notions of Unique Determination (UD) within measurement frameworks as follows:
    
    \begin{definition} [UDP]
    	A pure state $\ket{\psi}$ is Uniquely Determined among Pure states (UDP) under the measurement framework $\mathbf{A}$ if, for any pure state $\ket{\phi}$,
        \begin{align*}
        \Vec{\mathcal{M}}_{\mathbf{A}}(\ket{\phi}\bra{\phi}) = \Vec{\mathcal{M}}_{\mathbf{A}}(\ket{\psi}\bra{\psi})
        \implies |\langle\psi|\phi\rangle|^2 = 1.
        \end{align*}
    \end{definition}
    
    \begin{definition} [UDA]
    	A pure state $\ket{\psi}$ is Uniquely Determined among All states (UDA) under the measurement framework $\mathbf{A}$ if, for any pure or mixed state, represented by the density matrix $\rho$,
        \begin{align*}
        \Vec{\mathcal{M}}_{\mathbf{A}}(\rho) = \Vec{\mathcal{M}}_{\mathbf{A}}(\ket{\psi}\bra{\psi})
        \implies \Tr(\ket{\psi} \bra{\psi} \rho) = 1.
        \end{align*}
    \end{definition}
    
    \subsection*{Determining Uniqueness via Optimization}

    The determination of whether a pure quantum state $\ket{\psi}$ is UDP or UDA within a measurement framework $\mathbf{A}$ can be formulated as optimization problems. We present these formulations for both UDP and UDA cases.
    
    For the UDP problem, we aim to minimize the fidelity between $\ket{\psi}$ and a variable pure state $\ket{\phi}$, with the constraint that $\ket{\phi}$ is normalized and matches the measurement outcomes of $\ket{\psi}$:
    \begin{align*}
    	\underset{\ket{\phi} \in \mathbb{C}^d}{\text{minimize }} \;\; & |\langle\psi|\phi\rangle|^2 
        \\
    	\text{subject to } \;
        & \langle \phi | \phi \rangle = 1 
        \\ & 
        \Vec{\mathcal{M}}_{\mathbf{A}}(\ket{\phi}\bra{\phi}) = \Vec{\mathcal{M}}_{\mathbf{A}}(\ket{\psi}\bra{\psi}).
    \end{align*}
    The optimization landscape is inherently non-convex due to the quadratic nature of the fidelity term $||\langle\psi|\phi\rangle||^2$. If the optimal solution $\ket{\phi^*}$ satisfies $||\langle\psi|\phi\rangle||^2 = 1$, then $\ket{\psi}$ is confirmed to be UDP under $\mathbf{A}$. 
    
    Similarily, to determine whether a pure state $\ket{\psi}$ is UDA under $\mathbf{A}$, we cast this inquiry as a optimization problem over a density matrix $\rho$. The objective is to minimize the fidelity between the $\ket{\psi}$ and $\rho$ subject to semi-definite constraints:
    \begin{align*}
        \underset{\rho \in \mathbb{C}^{d \times d}}{\text{minimize }} \;\; & \Tr ( \ket{\psi} \bra{\psi} \rho ) \\
        \text{subject to } \;
        & \rho \succeq 0
        \\ &
        \Tr ( \rho ) = 1
        \\ &
        \Vec{\mathcal{M}}_{\mathbf{A}}(\rho) = \Vec{\mathcal{M}}_{\mathbf{A}}(\ket{\psi}\bra{\psi}).
    \end{align*}
    These constraints ensure that $\rho$ satisfies the properties of a valid quantum state while reproducing identical measurement outcomes to $\ket{\psi}$ under the framework $\mathbf{A}$. The convex nature of the UDA optimization problem allows for the application of SDP techniques, which are well-supported by numerous optimization software packages~\cite{majumdar2020recent,diamond2016cvxpy,mosek,scs,ocpb16scs}.
    Let $\rho^*$ denote the optimal solution to this problem. The state $\ket{\psi}$ is considered UDA under $\mathbf{A}$ if $\Tr (\ket{\psi} \bra{\psi} \rho^* ) = 1$.

    \subsection*{Rank and Decompositions of Density Matrices}
    
    The rank of a density matrix $\rho$, denoted as $\operatorname{rank}(\rho)$, is the number of non-zero eigenvalues of $\rho$. Given that $\rho$ is a positive semi-definite operator on a Hilbert space $\mathcal{H}^d$, its spectral properties can be thoroughly described by its eigenvalues and eigenvectors. A density matrix $\rho$ with a rank of $k$ can be succinctly expressed using spectral decomposition $\rho = \sum\nolimits_{i=1}^{k} \lambda_i |\psi_i\rangle \langle \psi_i|$, where $\ket{\psi_i}$ are the orthonormal eigenvectors of $\rho$ associated with the positive eigenvalues $\lambda_i$.

    The spectral decomposition is a special case of a more general representation known as the ensemble representation. In this representation, a density matrix $\rho$ whose rank is bounded by $r$ can be expressed by the ensemble of pure states:
    \begin{equation}
        \rho = \sum\nolimits_{i=1}^{r} p_i |\phi_i\rangle \langle \phi_i|,
    \end{equation}
    where ${\ket{\phi_i}}$ are not necessarily orthogonal, $p_i \geq 0$ are probabilities satisfying $\sum_{i=1}^r p_i = 1$. Using this representation, the optimization problems for UDP and UDA can be reformulated as:
    \begin{align*}
        \underset{ \{ \ket{\phi_i} \in \mathbb{C}^{d}, \: p_i \in \mathbb{R} \} }{\text{minimize}} & \Tr ( \ket{\psi} \bra{\psi} \rho ) 
        \\
    	\text{subject to } \;
    	& 
        \rho = \sum\nolimits_{i=1}^{r} p_i |\phi_i\rangle \langle \phi_i|
        \\ &
        \langle \phi_i | \phi_i \rangle = 1 
        \text{ and } p_i \ge 0
        \text{ for } i=1,\dots,r
        \\ &
        \sum\nolimits_{i=1}^{r} p_i = 1
        \\ &
        \Vec{\mathcal{M}}_{\mathbf{A}}(\rho) = \Vec{\mathcal{M}}_{\mathbf{A}}(\ket{\psi}\bra{\psi}).
    \end{align*}
    In this formulation, $r$ represents the maximum rank allowed for the variable density matrix $\rho$. For addressing UDP, we set $r=1$, corresponding to pure states. For addressing UDA, $r$ can be directly set to the dimension $d$ of the Hilbert space, to cover all pure states and mixed states. However, in the following subsection, we will demonstrate that it is not always necessary to set $r$ as high as $d$ when solving UDA problems, thereby potentially reducing computational complexity in certain cases.

    \subsection*{Low-Rank Solution in Quantum State Tomography}

    In QST, one seeks to reconstruct a quantum state $\rho$ from measurement data obtained via a set of observables $\mathbf{A} = \{ A_1, A_2, \dots, A_m \}$. Reconstructing $\rho$ exactly often requires considering all possible density matrices, which can have rank up to the dimension $d$ of the Hilbert space $\mathcal{H}^d$. To alleviate this complexity, we can exploit the fact that it is sufficient to consider density matrices of lower rank without loss of generality. Building upon the work of Chen et al.~\cite{chen2012rank}, we demonstrate that for any density matrix $\rho$, there exists a low-rank density matrix $\sigma$ that reproduces the same measurement outcomes under $\mathbf{A}$.

    \begin{theorem} [Existence of Low-Rank Solution in QST]
    \label{theo}
        Given a measurement framework $\mathbf{A} = \{A_1, A_2, ..., A_{m}\}$, for any density matrix $\rho$, there exists a density matrix $\sigma$ with rank less than $\sqrt{m+2}$ such that 
        \begin{equation*}
            \Vec{\mathcal{M}}_{\mathbf{A}}(\sigma) = \Vec{\mathcal{M}}_{\mathbf{A}}(\rho).
        \end{equation*}
    \end{theorem}
    
    \begin{proof}
    Let $\rho$ have spectral decomposition:
    \begin{equation*}
        \rho = \sum\nolimits_{i=1}^{r} \lambda_i |\psi_i\rangle \langle \psi_i|,
    \end{equation*}
    where $r = \operatorname{rank} (\rho)$ and $|\psi_i\rangle$ are the eigenvectors of $\rho$ associated with the positive eigenvalues ${\lambda_i}$. These eigenvectors form an orthonormal set spanning the support of $\rho$, denoted $\operatorname{supp}\rho$. Consider $B(\operatorname{supp}\rho)$, the set of bounded linear operators on the support of $\rho$. Particularly, we focus on its Hermitian subset, which is parametrized by $r^2$ real parameters.
    
    Extending the measurement framework to $\mathbf{B} = \mathbf{A} \cup \{ I \}$, which comprises $m+1$ observables, we define a subspace $M$ of Hermitian operators in $B(\operatorname{supp}\rho)$:
    \begin{equation*}
        M = \{ H \in B(\operatorname{supp}\rho): 
        \Vec{\mathcal{M}}_{\mathbf{B}}(H) = \Vec{0}
        \text{ and }
        H=H^{\dagger}
        \}.
    \end{equation*}
    The dimension of $M$ is at least $r^2 - (m + 1)$.
    
    If $r^2 - (m + 1) \ge 1$, $M$ is non-empty. For any non-zero $H \in M$, define $\sigma = \rho + \epsilon H$ for a small $\epsilon$. Since $\Vec{\mathcal{M}}_{\mathbf{B}}(H) = \Vec{0}$, we can see that the measurement vector and trace of $\sigma$ are unchanged:
    \begin{equation*}
    \Vec{\mathcal{M}}_{\mathbf{A}}(\sigma) = \Vec{\mathcal{M}}_{\mathbf{A}}(\rho) \text{ and } \Tr(\sigma) = \Tr(\rho) = 1.
    \end{equation*}

    Since the Hermitian operator $H$ with $\Tr(H) = 0$ contains both positive and negative eigenvalues, the same holds for $\rho + \lambda H$ for $\lambda \gg 1$. Hence, there exists an intermediate value $\epsilon$ for which $\sigma = \rho + \epsilon H$ is non-negative, but not strictly positive. In this case,
    $\sigma$ is a density matrix with $\operatorname{rank} (\sigma) < r$.
    
    By iteratively applying this procedure, reducing the rank at each step, we eventually reach a density matrix $\sigma$ with rank $r'$ such that $r'^2 - (m+1) < 1$. Thus, we obtain a density matrix $\sigma$ with $\operatorname{rank}(\sigma) < \sqrt{m+2}$ that satisfies $\Vec{\mathcal{M}}_{\mathbf{A}}(\sigma) = \Vec{\mathcal{M}}_{\mathbf{A}}(\rho)$.
    \end{proof}

    This theorem implies that when reconstructing a density matrix from measurement data, it suffices to consider density matrices of rank less than $\sqrt{m+2}$, where $m$ is the number of observables in the measurement framework.

    Conventional approaches to solving the UDA problem typically involve optimization over the entire space of density matrices, with ranks potentially as high as the Hilbert space dimension. However, by leveraging Theorem~\ref{theo}, we can refine our approach by considering only density matrices of restricted rank. We establish the following corollary.
    
    \begin{corollary}[Existence of Equivalent Low-Rank Density Matrix in UDA Problem]
    \label{coro}
    Given a measurement framework $\mathbf{A} = \{A_1, A_2, ..., A_{m}\}$ and a pure state $\ket{\psi}$,  for any density matrix $\rho$, there exists a density matrix $\sigma$ with rank less than $\sqrt{m+3}$ such that 
    \begin{equation*}
        \Tr(\ket{\psi}\bra{\psi}\sigma) = \Tr(\ket{\psi}\bra{\psi}\rho) \,
        \text{ and } \,
        \Vec{\mathcal{M}}_{\mathbf{A}}(\sigma) = \Vec{\mathcal{M}}_{\mathbf{A}}(\rho).
    \end{equation*}
    \end{corollary}

    \begin{proof}
        Extend the measurement framework $\mathbf{A}$ by adding the observable $\ket{\psi}\bra{\psi}$ to form $\mathbf{A'} = \mathbf{A} \cup \{ \ket{\psi}\bra{\psi} \}$, which now contains $m + 1$ observables. Applying Theorem~\ref{theo} to this extended framework $\mathbf{A'}$, we conclude that for any density matrix $\rho$, there exists a density matrix $\sigma$ with rank less than $\sqrt{m + 3}$ such that
        \begin{equation*}
        \Vec{\mathcal{M}}_{\mathbf{A'}}(\sigma) = \Vec{\mathcal{M}}_{\mathbf{A'}}(\rho).
        \end{equation*}
        This implies that both $\Tr(\ket{\psi}\bra{\psi}\sigma) = \Tr(\ket{\psi}\bra{\psi}\rho)$ and $\Vec{\mathcal{M}}_{\mathbf{A}}(\sigma) = \Vec{\mathcal{M}}_{\mathbf{A}}(\rho)$ hold. Therefore, the corollary follows directly from Theorem~\ref{theo}.
    \end{proof}

    This corollary ensures that for any density matrix $\rho$ as a solution to the UDA problem for target state $\ket{\psi}$, there always exists an equivalent low-rank density matrix $\sigma$ with rank less than $\sqrt{m+3}$, where $m$ is the number of measurement observables. This low-rank density matrix $\sigma$ is equivalent to $\rho$, since it not only reproduces the same measurement outcomes but also maintains the fidelity with $\ket{\psi}$. Consequently, when solving UDA problems, we can restrict our search to a subset of low-rank density matrices, reducing the number of variables in optimization.

    \section{Methodology}
    \label{sec:method}
    
    \subsection*{Parameterization of Quantum States}

    To address UD problems through optimization, we parameterize the variable density matrix using an ensemble of pure states: $\rho = \sum_{i=1}^{r} p_i |\phi_i\rangle \langle \phi_i|$. The number of pure states in the ensemble is constrained by the maximum allowed rank $r$ for $\rho$. 
    
    To streamline the parameterization, we represent the ensemble of pure states using a matrix $ V \in \mathbb{C}^{d \times r} $, where each column corresponds to a non-normalized pure state:
    \begin{equation}
        V = \begin{pmatrix} 
            |\tilde{\phi_1}\rangle & |\tilde{\phi_2}\rangle & \dots & |\tilde{\phi_r}\rangle 
        \end{pmatrix}.
    \end{equation}
    Each pure state in the ensemble is then normalized as:
    \begin{equation}
        |\phi_i\rangle = \frac{|\tilde{\phi_i}\rangle}{\sqrt{\langle \tilde{\phi_i} | \tilde{\phi_i} \rangle}}.
    \end{equation}
    This approach eliminates the need for explicit normalization constraints during optimization. To incorporate the probabilities $p_i$, we introduce a vector $ \Vec{q} \in \mathbb{R}^r $ and probabilities $p_i$ are represented by
    \begin{equation}
        p_i = \frac{q_i^2}{\|\Vec{q}\|_2^2}.
    \end{equation}
    This formulation ensures that the probabilities are non-negative and sum to one.

   We may avoid explicitly constructing $\rho$ in numerical calculation. The fidelity $f$ with the target state $\ket{\psi}$ becomes a function of $V$ and $\Vec{q}$:
    \begin{equation}
        f(V, \Vec{q}) = \sum_{i=1}^{r}
        \left(
         \frac{q_i^2}{\|\Vec{q}\|_2^2} \,
        \frac{\| \langle \psi | \phi_i' \rangle \|^2}{\langle \phi_i' | \phi_i' \rangle } 
        \right).
    \end{equation}
    The measurement constraints are encapsulated by the vector $\Vec{g}$:
    \begin{equation}
        \Vec{g}(V, \Vec{q}) = \Vec{\mathcal{M}}_{\mathbf{A}}(\rho) - \Vec{\mathcal{M}}_{\mathbf{A}}(\ket{\psi}\bra{\psi}).
    \end{equation}
    Each component of $ \Vec{g} $ can be given by:
    \begin{equation}
        g_j(V, \Vec{q}) = 
        \sum_{i=1}^{r} \left( \frac{q_i^2}{\|\Vec{q}\|_2^2} \, \frac{\langle \phi_i' | A_j | \phi_i' \rangle }{ \langle \phi_i' | \phi_i' \rangle } 
        \right)
        - \langle \psi | A_j | \psi \rangle.
    \end{equation}

    % In practical scenarios, measurement data may contain statistical errors. To account for this, it is essential to introduce relaxations for the equality constraints in optimization \cite{zambrano2024certification}. We relax the equality constraints for measurement outcomes by allowing a small deviation $\epsilon$:
    % \begin{equation}
    %    \| \Vec{g} \|_{\infty} = \| \Vec{\mathcal{M}}_{\mathbf{A}}(\rho) - \Vec{\mathcal{M}}_{\mathbf{A}}(\ket{\psi}\bra{\psi}) \|_{\infty} < \epsilon.
    % \end{equation}
    
    The optimization problem is thus formulated to minimize the fidelity while ensuring the measurement outcomes match our requirements:
    \begin{align*}
    \underset{
        V \in \mathbb{C}^{d \times r}, \: \Vec{q} \in \mathbb{R}^r
    }{\text{minimize }} \;\; 
    & f(V, \Vec{q}) \\
    \text{subject to } \;\;\;
    & \Vec{g}(V, \Vec{q})  = \Vec{0}.
    \end{align*}
    
    \subsection*{Numerical Approach to Unique Determinedness}

    When solving UDP problems, $r$ is fixed to be $1$ for optimization among pure state. For UDA determination, $r$ can be set to the largest integer less than $\sqrt{m+3}$ where $m$ is the number of non-identity observables in $\mathbf{A}$. This challenge is non-convex when $r$ is less than the Hilbert space dimension $d$.
    The objective function $f(V, \Vec{q})$ is the fidelity between target state $\ket{\psi}$ and the density matrix determined by $V$ and $\Vec{q}$ . It will be minimized with the constraint that variable density matrix must comply with the given measurement outcomes.
    
    The Augmented Lagrangian Method (ALM) is a powerful technique for solving constrained optimization problems, regardless of the convexity of the problem~\cite{burer2003nonlinear}. It combines principles of the penalty method with Lagrange multipliers to effectively handle equality and inequality constraints. In the context of quantum state tomography, ALM can be utilized to solve both UDP and UDA problems. In this approach, a constrained problem can be reformulate into a unconstrained one, thereby facilitating the application of gradient-based optimization techniques. By combining ALM and gradient-based techniques, one can effectively enforce the necessary constraints and efficiently achieve optimal solutions, even in high-dimensional Hilbert spaces of quantum systems.

    \begin{algorithm}
    \label{alg}
    \caption{State Uniqueness Determination}

    \KwIn{Measurement framework $\mathbf{A}$, 
    target state $\ket{\psi}$, maximum rank $r$,
    scaling factor $\gamma$, uniqueness threshold $\delta$, 
    initial value of penalty parameter $\mu_0$,
    maximum of penalty parameter $\mu_{\text{max}}$,
    and minimum of the weights $\alpha_{\text{min}}$ and $\beta_{\text{min}}$}
    
    \KwOut{Whether $\ket{\psi}$ is uniquely determined under $\mathbf{A}$}

    \BlankLine
    \BlankLine
    
    Randomly initialize $V \in \mathbb{C}^{d \times r}$ and $\Vec{q} \in \mathbb{R}^{r}$ independently from normal distribution $\mathcal{N}(0, 1)$;
    \\
    Normalize $\Vec{q}$ and each pure state $\ket{\phi_i}$ in $V$;
    \\
    Initialize $\Vec{\lambda} \leftarrow \Vec{0}$, $\mu \leftarrow \mu_0$, $\alpha \leftarrow 1$, $\beta \leftarrow 1$;

    \BlankLine
    \BlankLine
    
    \Repeat{convergence criteria are met}{
        Minimize $\mathcal{L}$ over $V$ and $\Vec{q}$ using Adam optimizer;
        \\
        Calculate $\Vec{g}$ according to the optimized $V$ and $\Vec{q}$;
        \\
        \uIf{$\gamma \, \mu \leq \mu_{\text{max}}$}{
            Update $\Vec{\lambda} \leftarrow \Vec{\lambda} + \mu \, \Vec{g}$, 
            $\mu \leftarrow \text{min}(\gamma \, \mu, \, \mu_\text{max})$
        }
        \uElseIf{$\alpha / \gamma \ge \alpha_{\text{min}}$}{
            Update $\Vec{\lambda} \leftarrow (\Vec{\lambda} + \mu\, \Vec{g})/\gamma$, 
            $\alpha \leftarrow \text{max}(\alpha / \gamma, \, \alpha_\text{min})$
        }
        \Else{
            Update $\Vec{\lambda} \leftarrow (\Vec{\lambda} + \mu\, \Vec{g})/\gamma$, 
            $\beta \leftarrow \text{max}(\beta / \gamma, \, \beta_\text{min})$
        }
    }

    \BlankLine
    \BlankLine
    
    \eIf{$f(V, \Vec{q}) > (1 - \delta)$}{
        \Return $\ket{\psi}$ is uniquely determined under $\mathbf{A}$;
    }{
        \Return $\ket{\psi}$ is not uniquely determined under $\mathbf{A}$.
    }
    
    \end{algorithm}

    In our implementation of ALM, the augmented objective function $\mathcal{L}$, which incorporates the original objective function and penalties for constraint violations, is given by:
    \begin{equation}
        \mathcal{L} 
        = \alpha \, f 
        + \beta \, \langle \Vec{\lambda} ,
           \Vec{g} \rangle
        + \frac{\mu}{2} \, \| \Vec{g} \| _{2}^{2}
    \end{equation}
    where $\alpha$ and $\beta$ are the weights for the first two terms,  $\mu$ is a positive scalar penalty parameter, $\Vec{\lambda}$ is a vector of Lagrange multipliers, $\Vec{g}$ encapsulates the equality constraints and indicates the infeasibility of the variables.

    To numerically solve the problem reformulated in ALM, we employ an iterative scheme that combines the Adam optimizer with ALM. Algorithm~\ref{alg} outlines the process for determining the uniqueness of a pure state $\ket{\psi}$ within a measurement framework $\mathbf{A}$, among density matrices whose ranks are bounded by $r$. The algorithm includes a maximum value $\mu_{\text{max}}$ to prevent the unbounded growth of the penalty parameter $\mu$ and the Lagrange multipliers $\Vec{\lambda}$, ensuring numerical stability. When $\mu$ reaches $\mu_{\text{max}}$, the algorithm adapts by scaling down the weights $\alpha$ and $\beta$ in turn. This strategic adjustment facilitates the convergence to feasible solutions while avoiding excessively large $\mu$ and $\Vec{\lambda}$.

    \section{Numerical Experiments}
    \label{sec:numer}
    
    For assessing the unique determinacy of a target pure state $\ket{\psi}$ within a measurement framework $\mathbf{A}$, fidelity serves as the principal metric. The criterion for unique determinacy is expressed as:
    \begin{equation}
    f(V^*, \Vec{q^*}) > (1-\delta).
    \label{eq:unique}
    \end{equation}
    where $V^*$ and $\Vec{q^*}$ are the optimal variables obtained by ALM. We establish a uniqueness threshold $\delta$ of $0.01$. That is, only if the corresponding fidelity exceeds $0.99$, $\ket{\psi}$ is considered uniquely determined among the density matrices whose ranks are bounded by a
    maximum rank $r$. For UDP problems, this maximum rank $r$ is set to $1$. For UDA problems, $r$ is determined based on Corollary~\ref{coro} that limits the rank of variable density matrix according to the number of observables in $\mathbf{A}$.

    %To account for statistical errors, a tolerance parameter $\epsilon$ of $10^{-6}$ is used. 
        
    From the perspective of numerical optimization, we allow a tolerance of $10^{-6}$ for the equality constraint, and all constraints are considered satisfied if $|\Vec{g}|_\infty < 10^{-6}$. The iterative procedure in our approach continues until two criteria are met: (1) all constraints are satisfied, and (2) the change in fidelity between successive iterations falls below $10^{-8}$. If convergence cannot be achieved with constraints satisfied for a particular target pure state, the ALM procedure may be restarted with a different random initialization of the variables. 
    
    In non-convex optimization, solutions obtained from ALM may converge to local optima rather than global ones. To ensure the reliability of our results, we verify the uniqueness properties through consistent reproduction. A target state $\ket{\psi}$ is considered uniquely determined under $\mathbf{A}$ only if the optimal solution consistently achieves a fidelity exceeding $1-\delta$ across $5$ different initializations. For efficiency, the process can be terminated early once a feasible solution is found with a fidelity below $1-\delta$, which suffices to disprove the unique determination of the target state. 

    \subsection*{Analysis of Qutrits}

     \begin{figure}
    	\includegraphics[width=0.45\textwidth]{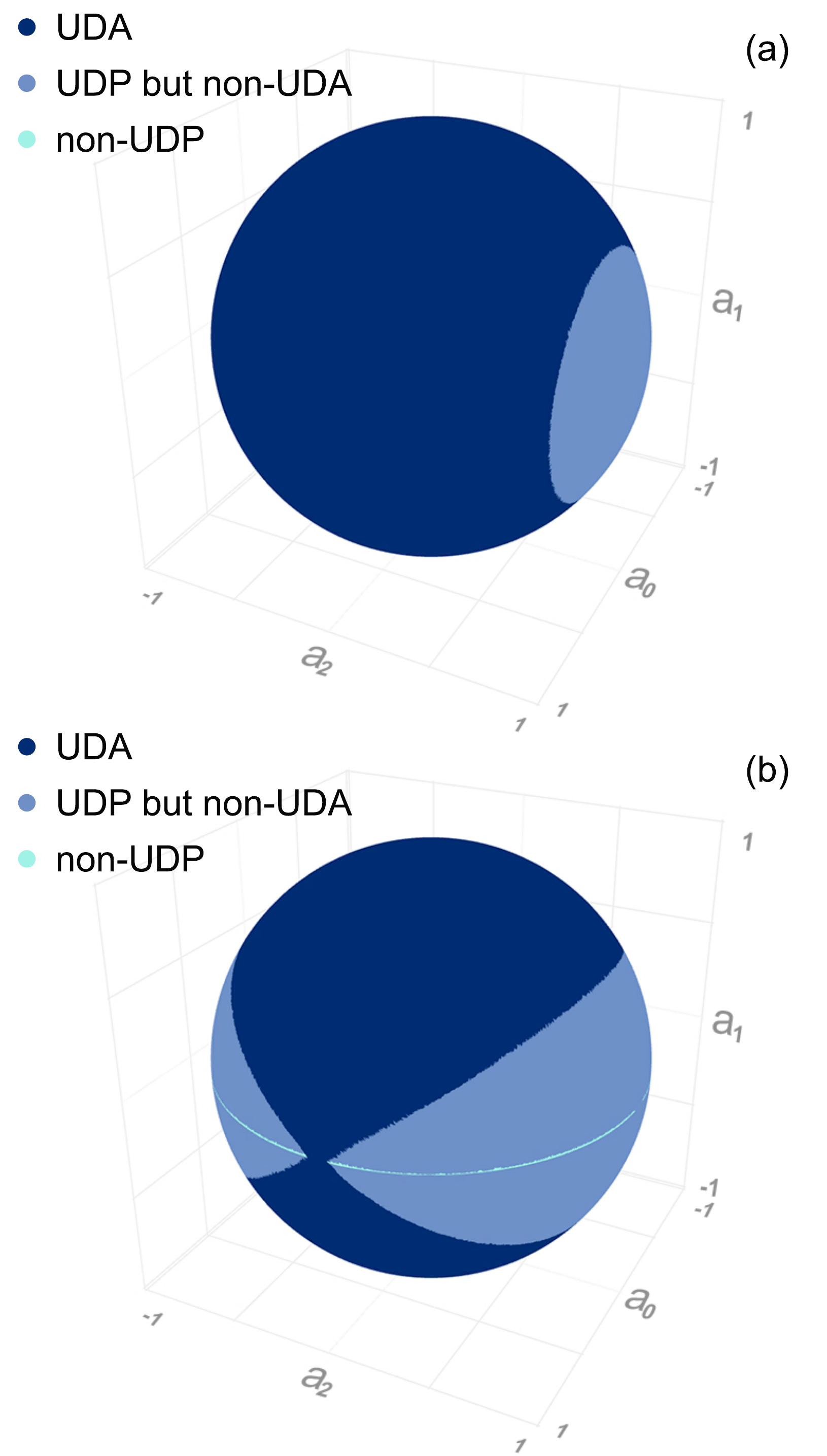}
    	\caption{Distribution of UDA and UDP cases under (a) $\mathbf{A_7}$ and (b) $\mathbf{A_6}$ on the surface of a unit ball with $a_1$ on the vertical axis and $a_0$, $a_2$ on the horizontal axes}
    	\label{fig:qutrit}
    \end{figure}
    
    We examine qutrit states in a three-dimensional Hilbert space, expressed as:
    \begin{equation}
    \ket{\psi} = a_0\ket{0} + a_1\ket{1} + a_2\ket{2},
    \end{equation}
    where the complex coefficients satisfy the normalization condition $|a_0|^2 + |a_1|^2 + |a_2|^2 = 1$. To investigate the uniqueness properties of qutrit states, the complete set of Gell-Mann matrices ${M_1, \ldots, M_8}$, along with the identity matrix, serve as the measurement observables. We consider three measurement frameworks: $\mathbf{A_8} = \{ M_1, M_2, \ldots, M_8 \}$ which includes all 8 Gell-Mann matrices, $\mathbf{A_7} = \mathbf{A_8} \setminus \{ M_8 \}$,  and $\mathbf{A_6} = \mathbf{A_8} \setminus \{ M_8, M_4 \}$ where 
    \begin{equation}
    		M_4 =
    		\begin{pmatrix}
    			0 & 0 & 1 \\
    			0 & 0 & 0 \\
    			1 & 0 & 0 \\
    		\end{pmatrix}
            \quad
            \text{and}
            \quad
            M_8 = \frac{1}{\sqrt{3}}
    		\begin{pmatrix}
    			1 & 0 & 0 \\
    			0 & 1 & 0 \\
    			0 & 0 & -2 \\
    		\end{pmatrix}.
            \label{num:M4M8}
    \end{equation}
    
    Framework $\mathbf{A_8}$, containing all Gell-Mann matrices, can uniquely determine any pure or mixed state, as it is informationally complete. When some observables are removed, the outcomes for UDP and UDA may differ. In such scenarios, some states may be UDP but not UDA. As dicussed in Appendix~\ref{apd:qutrit}, all pure states are UDP under $\mathbf{A_7}$. Under the framework $\mathbf{A_6}$, almost all pure states are UDP, except when $a_1 = 0$ and $|a_2| \in (0,1)$.
    
    To visualize our classification results, we represent states on the surface of a unit ball by setting $a_0$, $a_1$, and $a_2$ to be real. We sampled $500,000$ pure qutrit states on this surface and an additional $60,000$ states on the circle where $a_1=0$. Using the ALM approch, we confirmed that $100\%$ of the sampled states were UDA under $\mathbf{A_8}$. 
    
    As Figure~\ref{fig:qutrit} indicates, when applying the same method to framework $\mathbf{A_7}$, all sampled pure states, including those with $a_1 = 0$, remained UDP. However, around $20\%$ of sampled states were no longer UDA after the removal of $M_8$.
    
    By further reducing the framework and removing $M_8$ and $M_4$, the measurement framework $\mathbf{A_6}$ retains six non-identity observables. According to Corollary~\ref{coro}, the rank of the variable density matrix for UDA determination is theoretically bounded by 2. Among the 500,000 states sampled on the unit ball surface, $37\%$ were found not UDA, while all were classified as UDP under $\mathbf{A_6}$. For states sampled along the circle where $a_1=0$, UDA cases occurred near $\pm \ket{0}$, while UDP but not UDA cases appeared near $\pm \ket{2}$. All other states were classified as non-UDP under $\mathbf{A_6}$, which shows that the UDP classification results consistent with the theoretical analysis. 
    
    In these cases, the UDA classification was solved as a non-convex problem. To investigate the robustness, the numerical results from ALM were compared with those from SDP, known for global optimality. The comparison revealed minimal discrepancies, where the probabilities of being different in these cases are all less than $2 \times 10^{-5}$, underscoring the effectiveness of our ALM approach.
    
    \subsection*{Analysis of Four-Qubit Symmetric States}
    
    We now investigate symmetric states within four-qubit systems. These states can be expressed as superpositions of five symmetric basis states:
    \begin{equation}
    		|\phi\rangle 
    		= b_0|\omega_0\rangle + b_1|\omega_1\rangle + b_2|\omega_2\rangle + b_3|\omega_3\rangle + b_4|\omega_4\rangle
    \end{equation}
    where $|\omega_i\rangle$ represents the normalized symmetric state with $i$ qubits in the $|1\rangle$ state and $(4-i)$ qubits in the $|0\rangle$ state. Here, we are particularly interested in the uniqueness properties of four-qubit symmetric states by their two-particle reduced density matrices (2-RDMs). The symmetry of these states implies identical 2-RDMs, regardless of which pair of qubits is traced out.

    \subsubsection*{Analysis of Generalized GHZ States}

    \begin{figure}
        \centering
    \includegraphics[width=0.48\textwidth]{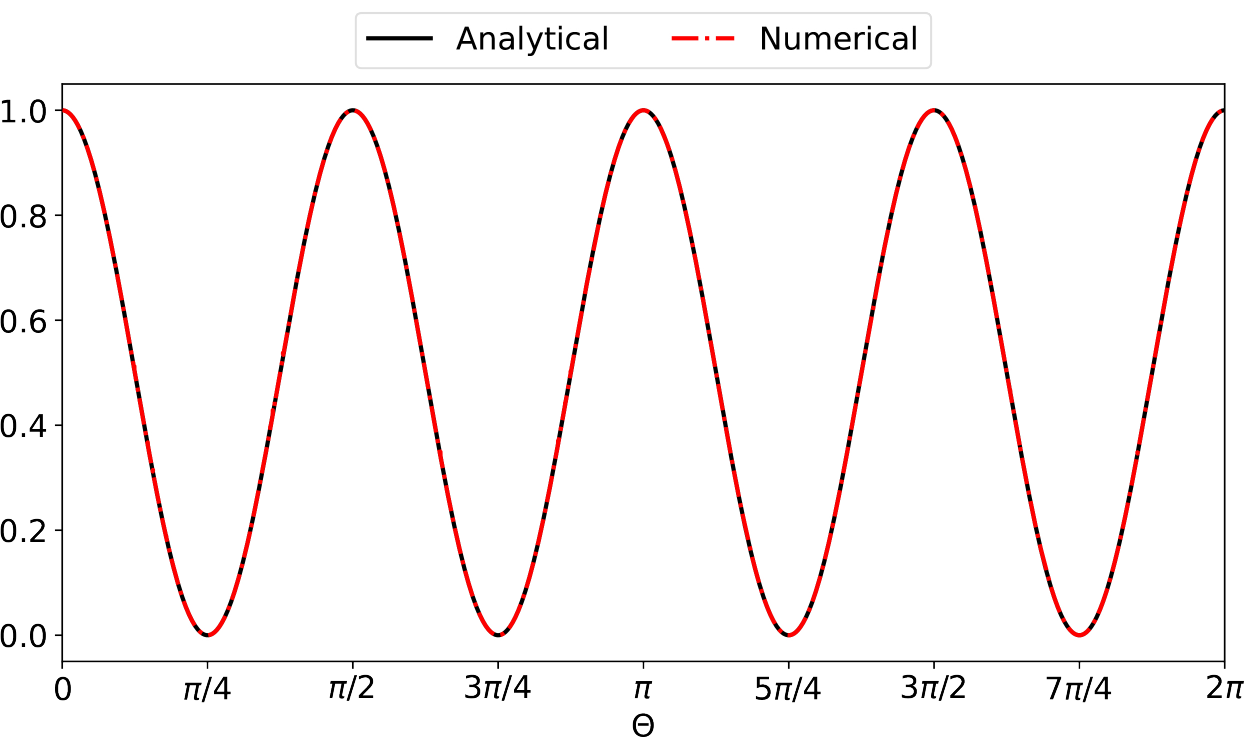}
    	\caption{The minimal fidelity of UDP optimization problems for four-qubit generalized GHZ states, which are represented as $|\psi_\text{GHZ} \rangle = \sin{\Theta}|\omega_0\rangle + \cos{\Theta}|\omega_4\rangle$.}
    	\label{fig:ghz}
    \end{figure}

    We begin our analysis with the subspace of these symmetric states spanned by $\ket{\omega_0}=\ket{0000}$ and $\ket{\omega_4}=\ket{1111}$. States within this subspace are recognized as a generalized Greenberger-Horne-Zeilinger (GHZ) state, which extends the concept of the standard GHZ state to include arbitrary complex coefficients. Ignoring the complex phase, we may parameterize such a state using a angle $\Theta$:
    \begin{equation}
    	|\psi_\text{GHZ} \rangle 
    	= \sin{\Theta}|\omega_0\rangle + \cos{\Theta}|\omega_4\rangle.
    	\label{num:eq:ghz}
    \end{equation}
    As shown in Appendix~\ref{apd:sub:ghz}, these generalized GHZ states are not UDP by their 2-RDMs except when $\Theta = k\,\pi/2$ with integer $k$, corresponding to $\pm\ket{\omega_0}$ or $\pm\ket{\omega_4}$. For this case, the global optimal solution of the UDP optimization problem can be expressed analytically with $\Theta$:
    \begin{equation}
    	|\phi_\text{GHZ}^* \rangle 
    	=  \sin{\Theta}|\omega_0\rangle - \cos{\Theta}|\omega_4\rangle,
     \end{equation}
     along with the corresponding minimal fidelity 
     \begin{equation}
        f^*_\text{GHZ} = \| \langle \psi_\text{GHZ} | \phi^*_\text{GHZ} \rangle \|^2 = \cos^2{2 \Theta}.
    	\label{num:sol:ghz}
    \end{equation}
    To validate our ALM approach, we compare its numerical results with these analytical solutions. We sampled $15，000$ pure states by varying the value of $\Theta$. As Figure~\ref{fig:ghz} illustrates, our numerical results for four-qubit generalized GHZ states can closely reproduce the analytical results, with a mean square error of minimal fidelity less than $5 \times 10^{-8}$.

    \subsubsection*{Analysis of Special Symmetric States}

    \begin{figure}
        \includegraphics[width=0.45\textwidth]{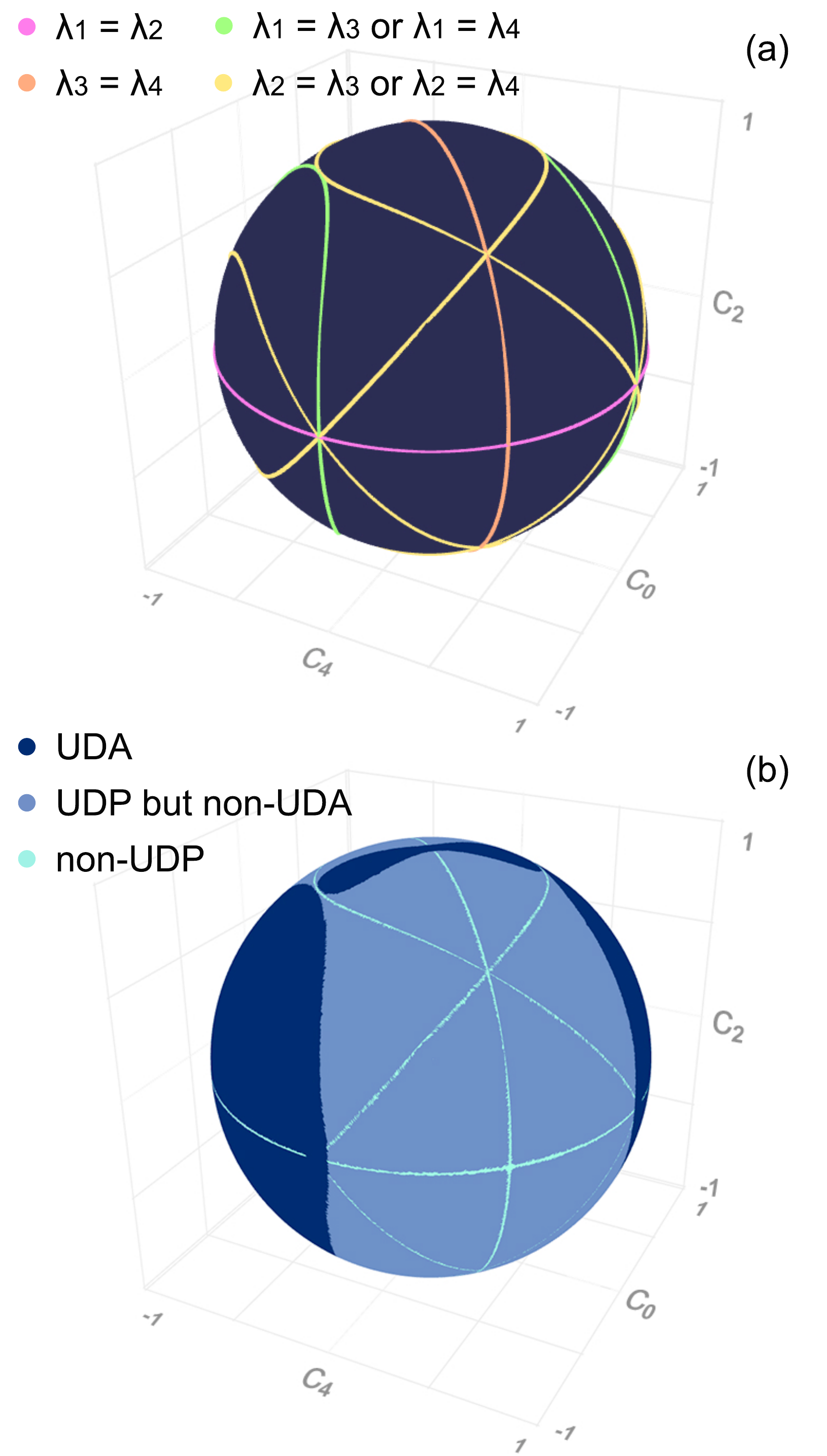}
    	\caption{(a) Parameter space of the symmetric state $\ket{\psi}$ in the form (\ref{num:eq:psi}), showing the colored curves corresponding to four types of 2-RDM eigenvalue degeneracies: (i) $\lambda_1 = \lambda_2$, (ii) $\lambda_3 = \lambda_4$, (iii) $\lambda_1 = \lambda_3$ or $\lambda_1 = \lambda_4$, and (iv) $\lambda_2 = \lambda_3$ or $\lambda_2 = \lambda_4$. (b) Classification of UDA and UDP cases across the parameter space of $\ket{\psi}$ in the form (\ref{num:eq:psi}), categorized into three groups: (A) UDA, (B) UDP but not UDA, and (C) neither UDP nor UDA.}
    	\label{fig:q4sym}
    \end{figure}

    We then expanded our investigation to a larger subset of four-qubit symmetric states, which can be represented as:
    \begin{equation}
    	|\psi\rangle 
    	= c_0|\omega_0\rangle + c_2|\omega_2\rangle + c_4|\omega_4\rangle,
    	\label{num:eq:psi}
    \end{equation}
    where $c_0$, $c_2$, and $c_4$ are real coefficients satisfying the normalization condition. Each 2-RDM for state $|\psi\rangle$ can be expressed as:
    \begin{equation}
    		\rho_{2,\psi} =
    		\begin{pmatrix}
    			c_0^2+\frac{c_2^2}{6} & 0 & 0 & \frac{c_2(c_0+c_4)}{\sqrt{6}} \\
    			0 & \frac{c_2^2}{3} & \frac{c_2^2}{3} & 0 \\
    			0 & \frac{c_2^2}{3} & \frac{c_2^2}{3} & 0 \\
    			\frac{c_2(c_0+c_4)}{\sqrt{6}} & 0 & 0 & c_4^2+\frac{c_2^2}{6} \\
    		\end{pmatrix}
    		\label{num:eq:psi_rdm}
    \end{equation}
    
    As established in previous research~\cite{PhysRevLett.118.020401}, a state $\ket{\psi}$ in the form (\ref{num:eq:psi}) is UDP by its 2-RDMs provided that all eigenvalues of its 2-RDMs are non-degenerate. However, the uniqueness remains unresolved when eigenvalue degeneracies are present.
        
    The eigenvalues of the 2-RDM in the form (\ref{num:eq:psi_rdm}) can be expressed in terms of $c_0$, $c_2$, and $c_4$:
    \begin{equation}
    	\lambda_1 = 0, \quad \lambda_2 = \frac{2 c_2^2}{3},
    	\label{eq:lamb1and2}
    \end{equation}
    \begin{equation}
    	\lambda_{3,4} = \frac{(c_0^2+c_4^2+\frac{c_2^2}{3}) \pm \sqrt{(c_0+c_4)^2((c_0-c_4)^2+\frac{2c_2^2}{3})}}{2}.
    	\label{eq:lamb3and4}
    \end{equation}
    For the 2-RDM in the form (\ref{num:eq:psi_rdm}), as  indicated by the colored curves in Figure~\ref{fig:q4sym}(a), eigenvalue degeneracies occur in four distinct ways: (i) $\lambda_1 = \lambda_2$, (ii) $\lambda_3 = \lambda_4$, (iii) $\lambda_1 = \lambda_3$ or $\lambda_1 = \lambda_4$, and (iv) $\lambda_2 = \lambda_3$ or $\lambda_2 = \lambda_4$. These degeneracy cases are visualized as colored curves on the parameter space of $\ket{\psi}$ in Figure \ref{fig:q4sym}(a). Each curve corresponds to a set of states where a specific type of eigenvalue degeneracy occurs.
    
    To determine the uniqueness of a four-qubit pure state by its 2-RDMs, we included all single and two-qubit Pauli operators in our measurement framework. This framework contains $66$ non-identity operators, and by Corollary~\ref{coro}, the rank of the variable density matrices for UDA problems can be restricted to at most $8$. However, according to Appendix~\ref{apd:sub:rank}, the specific properties of the 2-RDMs in the form (\ref{num:eq:psi_rdm}) allow us to further reduce the maximum rank of the variable density matrices to $5$.
        
    To explore the uniqueness properties across the full subset, we performed extensive sampling for the non-degenerate cases and degenerate cases separately. For non-degenerate cases, we randomly sampled $500,000$ 4-qubit symmetric states throughout parameter space of $\ket{\psi}$. For each of the four degeneracy types, we sampled $15,000$ states along their respective degeneracy curves on the unit sphere. The classification results are presented in Figure~\ref{fig:q4sym}(b), categorizing the states into three distinct groups: (A) UDA, (B) UDP but not UDA, and (C) not UDP. Our findings indicate that despite the presence of degeneracies, many symmetric states remain UDP, and some are even UDA by their 2-RDMs.

    In this case, the rank of variable density matrix is bounded by $5$ for UDA problems. To assess the effectiveness of our non-convex ALM approach, we compared our results with those obtained by SDP methods. In practice, SDP may sometimes struggle to find the feasible solution for UDA problems, when target states are located near the intersection points of degeneracy cases (iii) $\lambda_1 = \lambda_3$ and $\lambda_1 = \lambda_4$ and (iv) $\lambda_2 = \lambda_3$ and $\lambda_2 = \lambda_4$. For example, one intersection point of (iii) and (iv) corresponds to the target pure state 
    \begin{equation}
        \ket{\psi_\text{inter}} = \frac{1}{2\sqrt{2}} \ket{\omega_0} + \frac{\sqrt{3}}{2} \ket{\omega_2} + \frac{1}{2\sqrt{2}} \ket{\omega_4},
    \end{equation}
    which shares identical 2-RDMs with another symmetric pure state
    \begin{equation}
        \ket{\phi_\text{inter}} = \frac{1}{\sqrt{2}} \ket{\omega_1} + \frac{1}{\sqrt{2}} \ket{\omega_3}.
    \end{equation}
    Our ALM approach identifies target pure states near these intersection points as non-UDA, with some even classified as non-UDP. Among cases where SDP converges accurately, the ALM approach achieves a accuracy exceeding $99.5\%$ in determining UDA.

    \quad

    \section{Conclusion}
    \label{sec:conclu}

    This paper presents a novel numerical approach for investigating the uniqueness properties of quantum states in tomography, in which UDP and UDA problems are formulated as optimization challenges. We also have proposed a theorem which shows the existence of low-rank solution in QST. This theorem allows us to reformulate the UDA problem with low-rank constraints, which can reduce number of variables in optimization but introduce non-convexity on the other hand. By leveraging the Augmented Lagrangian Method, we have developed a unified framework to handle the constrained non-convex optimization challenges, inherent in both UDP and low-rank UDA determinations.

    Through comprehensive numerical experiments on qutrit systems and four-qubit symmetric states, we validated our theoretical classifications and demonstrated the robustness of our approach. For qutrit systems, we illustrated the distribution of UDP and UDA states under various measurement frameworks, confirming consistency with theoretical predictions and analytical solutions. In analyzing four-qubit symmetric states, we focused on specific forms and their 2-RDMs. The close agreement between our ALM results and analytical solutions for generalized GHZ states, along with low discrepancies compared to SDP methods, underscores the accuracy and reliability of our approach. By visualizing the distribution of three UD categories, we provided clear insights into how eigenvalue degeneracies in RDMs affect the uniqueness properties of states. This enhances our understanding of the role that symmetries and degeneracies play in the reconstruction of quantum states from local measurements.
        
    Future research directions may include extending our methods to even larger and more complex quantum systems, exploring the impact of different types of noise and errors in measurements, and applying our approach to experimental data from actual quantum devices. Additionally, investigating alternative optimization techniques or further refining the low-rank reformulation could yield even more efficient solutions.

    \bibliographystyle{unsrt}
    \bibliography{ref.bib}

    \clearpage 
    % \onecolumngrid

    \appendix

    \section{Qutrits and Gell-Mann Matrices}
    \label{apd:qutrit}
    
    Qutrits are the three-level generalization of quantum bits (qubits), which are the fundamental units of quantum information. While a qubit has two possible states (typically denoted as $\ket{0}$ and $\ket{1}$), a qutrit has three orthogonal basis states, often represented as $\ket{0}$, $\ket{1}$  and $\ket{2}$.
    
    For qubits, QST utilizes the Pauli matrices, which form a complete basis for the space of $2 \times 2$ Hermitian matrices. In the case of qutrits, Gell-Mann matrices play a similar role as they form a complete basis for the space of $3 \times 3$ Hermitian matrices. The eight Gell-Mann matrices is usually labelled $M_1$ to $M_8$ and they ares:
    \begin{equation}
    	\begin{aligned}
    		M_1 &= \begin{pmatrix} 0 & 1 & 0 \\ 1 & 0 & 0 \\ 0 & 0 & 0 \end{pmatrix}, &
    		M_2 &= \begin{pmatrix} 0 & -i & 0 \\ i & 0 & 0 \\ 0 & 0 & 0 \end{pmatrix}, \\
    		M_3 &= \begin{pmatrix} 1 & 0 & 0 \\ 0 & -1 & 0 \\ 0 & 0 & 0 \end{pmatrix}, &
    		M_4 &= \begin{pmatrix} 0 & 0 & 1 \\ 0 & 0 & 0 \\ 1 & 0 & 0 \end{pmatrix}, \\
    		M_5 &= \begin{pmatrix} 0 & 0 & -i \\ 0 & 0 & 0 \\ i & 0 & 0 \end{pmatrix}, &
    		M_6 &= \begin{pmatrix} 0 & 0 & 0 \\ 0 & 0 & 1 \\ 0 & 1 & 0 \end{pmatrix}, \\
    		M_7 &= \begin{pmatrix} 0 & 0 & 0 \\ 0 & 0 & -i \\ 0 & i & 0 \end{pmatrix}, &
    		M_8 &= \frac{1}{\sqrt{3}} \begin{pmatrix} 1 & 0 & 0 \\ 0 & 1 & 0 \\ 0 & 0 & -2 \end{pmatrix}.
    	\end{aligned}
    \end{equation}

    For a general qutrit state represented as 
    \begin{equation}
        \ket{\psi} = a_0 \ket{0} + a_1 \ket{1} + a_2 \ket{2},
        \label{apd:qutrit:psi}
    \end{equation}
    where the complex coefficients $a_0$, $a_1$ and $a_2$ satisfy the normalization condition $|a_0|^2 + |a_1|^2 + |a_2|^2 = 1$, the expectation values of the Gell-Mann matrices are given by:
    \begin{equation}
        \begin{aligned}
            & \langle M_1 \rangle = 2 \, \text{Re}(a_0^* a_1), 
            & \langle M_2 \rangle & = 2 \, \text{Im}(a_0^* a_1), \\
            & \langle M_3 \rangle = |a_0|^2-|a_1|^2, 
            & \langle M_4 \rangle & = 2 \, \text{Re}(a_0^* a_2), \\
            & \langle M_5 \rangle = 2 \, \text{Im}(a_0^* a_2),
            & \langle M_6 \rangle & = 2 \, \text{Re}(a_1^* a_2), \\
            & \langle M_7 \rangle = 2 \, \text{Im}(a_1^* a_2),
            & \langle M_8 \rangle & = (|a_0|^2+|a_1|^2-2|a_2|^2)/\sqrt{3}.
        \end{aligned}
    \end{equation}

    When all coefficients $a_0$, $a_1$ and $a_2$ are non-zero, the phases of these complex numbers play a crucial role in the state's characterization. Specifically, the coefficients $a_0$, $a_1$, and $a_2$ share the same phase if and only if all the imaginary expectation values vanish simultaneously: $\langle M_2 \rangle = \langle M_5 \rangle = \langle M_7 \rangle = 0.$

    \begin{theorem}[UDP Property of Qutrit States]
        Suppose $\mathbf{A_8}$ is the measurement framework including all eight Gell-Mann matrices. Consider two reduced measurement frameworks $\mathbf{A_7} = \mathbf{A_8} \setminus \{M_8\}$ and $\mathbf{A_6} = \mathbf{A_8} \setminus \{M_4, M_8\}$. A qutrit state $\ket{\psi}=a_0 \ket{0} + a_1 \ket{1} + a_2 \ket{2}$ with $a_0, a_1, a_2 \in \mathbb{R}$ is always UDP under $\mathbf{A_7}$, and $\ket{\psi}$ is non-UDP under $\mathbf{A_6}$ if and only if $a_0=0$ and $a_1, a_2 \neq 0$.
    \end{theorem}

    \begin{proof}
    For a qutrit state $\ket{\psi}$ with real coefficients, the expectation values of the Gell-Mann matrices simplify to:
    \begin{align*}
        & \langle M_1 \rangle = 2 \, a_0 a_1, 
        & \langle M_2 \rangle & =\langle M_5 \rangle =\langle M_7 \rangle=0, \\
        & \langle M_3 \rangle = a_0^2-a_1^2, 
        & \langle M_4 \rangle & = 2 \, a_0 a_2, \\
        & \langle M_6 \rangle  = 2 \, a_1 a_2, 
        & \langle M_8 \rangle & = (a_0^2+a_1^2-2a_2^2)/\sqrt{3}.
    \end{align*}

    When $a_1 \neq 0$, $a_0^2$ and $a_1^2$ can be uniquely determined from $\langle M_1 \rangle$ and $\langle M_3 \rangle$:
    \begin{align*}
        a_1^2 = \frac{\sqrt{\langle M_1 \rangle^2+\langle M_3 \rangle^2}-\langle M_3 \rangle}{2},
        \\
        a_0^2 =\frac{\sqrt{\langle M_1 \rangle^2+\langle M_3 \rangle^2}+\langle M_3 \rangle}{2}.
    \end{align*}
    Using the normalization condition $a_0^2+a_1^2+a_2^2=1$, we obtain:
    \begin{align*}
        a_2^2 = 1 - \sqrt{\langle M_1 \rangle^2+\langle M_3 \rangle^2}.
    \end{align*}
    Without loss of generality, we may assume $a_1 > 0$, and the signs of $a_0$ and $a_2$ can be determined by the signs of $\langle M_1 \rangle$ and $\langle M_6 \rangle$, respectively.

    When $a_1 = 0$, we have $\langle M_1 \rangle = \langle M_6 \rangle = 0$, and:
    \begin{align*}
        a_0^2 = \langle M_3 \rangle, \quad a_2^2 = 1 - \langle M_3 \rangle.
    \end{align*}
    Under $\mathbf{A_7}$, the signs of $a_0$ and $a_2$ are determined by $\langle M_4 \rangle$, so $\ket{\psi}$ is still UDP. However, under $\mathbf{A_6}$ where $M_4$ is removed, the signs become indeterminate if both $a_0$ and $a_2$ are non-zero, making $\ket{\psi}$ non-UDP. For example, states $a_0\ket{0} + a_2\ket{2}$ and $a_0\ket{0} - a_2\ket{2}$ share the same measurement outcomes under $\mathbf{A_6}$. The only UDP states under $\mathbf{A_6}$ with $a_1=0$ are those with either $a_0=0$ or $a_2=0$.

    Therefore, any qutrit state $\ket{\psi}$ with real coefficients is UDP under $\mathbf{A_7}$. Under $\mathbf{A_6}$, $\ket{\psi}$ is UDP if and only if $a_1 \neq 0$ or $a_0, a_2 = 0$.
    \end{proof}

    \section{Four-Qubit Symmetric States and Shared Identical 2-RDMs}
    \label{apd:q4sym}

    In this section, we explore four-qubit symmetric states, which can be represented as linear combinations of five basis states:
    \begin{equation}
        \ket{\phi} 
        = b_0 \ket{\omega_0} 
        + b_1 \ket{\omega_1} 
        + b_2 \ket{\omega_2} 
        + b_3 \ket{\omega_3} 
        + b_4 \ket{\omega_4},
        \label{apd:phi}
    \end{equation}
    where $b_0, b_1, b_2, b_3, b_4 \in \mathbb{C}$ and $\ket{\omega_i}$ denotes the normalized symmetric state. The explicit forms of these basis states are:
    \begin{equation}
        \begin{aligned}
            \ket{\omega_0} &= \ket{0000}, 
            \\
            \ket{\omega_1} &= \frac{\ket{0001}+\ket{0010}+\ket{0100}+\ket{1000}}{2}, 
            \\
            \ket{\omega_2} &= \frac{\ket{0011}+\ket{0101}+\ket{0110}+\ket{1001}+\ket{1010}+\ket{1100}}{\sqrt{6}}, 
            \\
            \ket{\omega_3} &= \frac{\ket{0111}+\ket{1011}+\ket{1101}+\ket{1110}}{2}, 
            \\
            \ket{\omega_4} &= \ket{1111}.
        \end{aligned}
    \end{equation}

    Our objective is to investigate the uniqueness properties of four-qubit symmetric states based on their 2-RDMs. Due to their symmetry, these states possess identical 2-RDMs independent of which qubit pair is traced out. Consider a general four-qubit symmetric pure state $\ket{\psi}$ expressed as:
    \begin{equation}
        \ket{\psi} 
        = c_0 \ket{\omega_0} 
        + c_1 \ket{\omega_1} 
        + c_2 \ket{\omega_2} 
        + c_3 \ket{\omega_3} 
        + c_4 \ket{\omega_4},
        \label{apd:psi}
    \end{equation}
    where $c_0, c_1, c_2, c_3, c_4 \in \mathbb{C}$. The symmetric state $\ket{\psi}$ is non-UDP by its 2-RDMs if there exists a symmetric state $\ket{\phi}$ in the form (\ref{apd:phi}) such that they share the identical 2-RDMs and $|\langle\psi|\phi\rangle|^2<1$. For these two states sharing identical 2-RDMs, the following equality constraints must be satisfied:
    \begin{equation}
    \begin{aligned}
        & |b_0|^2+\frac{1}{2}|b_1|^2+\frac{1}{6}|b_2|^2=|c_0|^2+\frac{1}{2}|c_1|^2+\frac{1}{6}|c_2|^2,
        \\
        & |b_4|^2+\frac{1}{2}|b_3|^2+\frac{1}{6}|b_2|^2 = |c_4|^2+\frac{1}{2}|c_3|^2+\frac{1}{6}|c_2|^2,
    	\\
    	& \frac{1}{4}|b_1|^2+\frac{1}{4}|b_3|^2+\frac{1}{3}|b_2|^2 = \frac{1}{4}|c_1|^2+\frac{1}{4}|c_3|^2+\frac{1}{3}|c_2|^2,
    	\\
    	& \frac{1}{2}b_1^*b_0+\frac{1}{\sqrt{6}}b_2^*b_1+\frac{1}{2\sqrt{6}}b_3^*b_2=\frac{1}{2}c_1^*c_0+\frac{1}{\sqrt{6}}c_2^*c_1+\frac{1}{2\sqrt{6}}c_3^*c_2,
        \\
    	& \frac{1}{2}b_4^*b_3+\frac{1}{\sqrt{6}}b_3^*b_2+\frac{1}{2\sqrt{6}}b_2^*b_1 = \frac{1}{2}c_4^*c_3+\frac{1}{\sqrt{6}}c_3^*c_2+\frac{1}{2\sqrt{6}}c_2^*c_1,
        \\
    	& \frac{1}{\sqrt{6}}b_2^*b_0+\frac{1}{2}b_3^*b_1+\frac{1}{\sqrt{6}}b_4^*b_2 = \frac{1}{\sqrt{6}}c_2^*c_0+\frac{1}{2}c_3^*c_1+\frac{1}{\sqrt{6}}c_4^*c_2,
        \label{apd:same_2rdm}
    \end{aligned}
    \end{equation}
    along with the normalization conditions:
    \begin{equation}
        \begin{aligned}
        |b_0|^2+|b_1|^2+|b_2|^2+|b_3|^2+|b_4|^2 = 1,
        \\
        |c_0|^2+|c_1|^2+|c_2|^2+|c_3|^2+|c_4|^2 = 1.
        \end{aligned}
    \end{equation}

    \subsection{Uniqueness Property of Generalized GHZ States}
    \label{apd:sub:ghz}

     Generalized GHZ states are a subclass of symmetric states, which are obtained by setting the coefficients $c_1 = c_2 = c_3 = 0$ in the symmetric state expression (\ref{apd:psi}). Assuming $c_0, c_4 \in \mathbb{R}$, these coefficients can be parameterized as functions of a single angle $\Theta$:
    \begin{equation}
        c_0 = \sin{\Theta}, \quad c_4 = \cos{\Theta}.
    \end{equation}
    Thus, a generalized GHZ state with real coefficients can be expressed as:
    \begin{equation}
        \ket{\psi_\text{GHZ}} = \sin{\Theta} \ket{\omega_0} + \cos{\Theta} \ket{\omega_4}.
        \label{apd:ghz:psi}
    \end{equation}
    
    The uniqueness properties of these states based on their 2-RDMs can be analyzed using the following theorem.
    
    \begin{theorem}[Optimal Solution for UDP Problem of Generalized GHZ States]
    \label{apd:theo_ghz}
    For a generalized GHZ state $\ket{\psi_\text{GHZ}}$ in the form (\ref{apd:ghz:psi}), the global optimal solution to the UDP optimization problem can be expressed as:
    \begin{align*}
        |\phi_\text{GHZ}^* \rangle = \sin{\Theta}|\omega_0\rangle - \cos{\Theta}|\omega_4\rangle,
    \end{align*}
    with the corresponding minimal fidelity given by:
    \begin{align*}
        f^*_\text{GHZ} = \| \langle \psi_\text{GHZ} | \phi^*_\text{GHZ} \rangle \|^2 = \cos^2{2 \Theta}.
    \end{align*}
    \end{theorem}
    
    \begin{proof}
    For a generalized GHZ state $\ket{\psi_\text{GHZ}}$ in the form (\ref{apd:ghz:psi}), equations (\ref{apd:same_2rdm}) imply that any state $\ket{\phi}$ in the form (\ref{apd:phi}) sharing the same 2-RDMs with $\ket{\psi_\text{GHZ}}$ must satisfy:
    \begin{align*}
        |b_0|^2 = |c_0|^2, \quad b_1 = b_2 = b_3 = 0, \quad |b_4|^2 = |c_4|^2.
    \end{align*}
    Without loss of generality, we can parameterize $\ket{\phi}$ as:
    \begin{align*}
        b_0 = c_0, \quad b_4 = c_4 e^{i\gamma}.
    \end{align*}
    The fidelity between $\ket{\psi_\text{GHZ}}$ and $\ket{\phi}$ is then:
    \begin{align*}
        f = |\langle \psi_\text{GHZ} | \phi \rangle|^2 = |\sin^2{\Theta} + e^{i\gamma} \cos^2{\Theta}|^2.
    \end{align*}
    
    To minimize $f$, we take the partial derivative with respect to $\gamma$:
    \begin{align*}
        \frac{\partial f}{\partial \gamma} = -2 \cos^2{\Theta} \sin^2{\Theta} \sin{\gamma}.
    \end{align*}
    This derivative is zero when $\gamma = n \pi$ for integer $n$. When $n$ is even, $e^{i\gamma} = 1$, and the fidelity $f = 1$ achieves its maximum value. When $n$ is odd, $e^{i\gamma} = -1$, and the fidelity reaches its minimum:
    \begin{align*}
        \min_{\gamma} f = |\sin^2{\Theta} - \cos^2{\Theta}|^2 = \cos^2{2 \Theta}.
    \end{align*}
    The corresponding state $\ket{\phi_\text{GHZ}^*}$ that minimizes $f$ is:
    \begin{align*}
        |\phi_\text{GHZ}^* \rangle = \sin{\Theta}|\omega_0\rangle - \cos{\Theta}|\omega_4\rangle.
    \end{align*}
    This completes the proof.
    \end{proof}
    
    \begin{corollary}[UDP Property of Generalized GHZ States]
    \label{apd:coro_ghz}
    A generalized GHZ state $\ket{\psi_\text{GHZ}}$ in the form (\ref{apd:ghz:psi}) is UDP by its 2-RDMs if and only if $\ket{\psi_\text{GHZ}}$ is either $\pm \ket{\omega_0}$ or $\pm \ket{\omega_4}$.
    \end{corollary}
    
    \begin{proof}
    Based on Theorem~\ref{apd:theo_ghz}, the minimal fidelity of the UDP optimization problem for $\ket{\psi_\text{GHZ}}$ is $f^*_\text{GHZ} = \cos^2{2 \Theta}$. For $f^*_\text{GHZ} = 1$, it must hold that $\cos^2{2 \Theta} = 1$, which occurs if and only if $2 \Theta = k \pi$ for integer $k$. Therefore, $\Theta = k \pi / 2$, corresponding to $\ket{\psi_\text{GHZ}} = \pm \ket{\omega_0}$ or $\pm \ket{\omega_4}$. For all other values of $\Theta$, $f^*_\text{GHZ} < 1$, and $\ket{\psi_\text{GHZ}}$ is not UDP.
    \end{proof}
    
    Therefore, a generalized GHZ state $\ket{\psi_\text{GHZ}}$ in the form (\ref{apd:ghz:psi}) is not UDP by its 2-RDMs unless $\Theta = k \pi / 2$ for integer $k$. The optimal solution to the UDP optimization problem achieves a minimal fidelity of $\cos^2{2 \Theta}$, with the corresponding optimal state given by:
    \begin{equation}
        |\phi_\text{GHZ}^* \rangle = \sin{\Theta}|\omega_0\rangle - \cos{\Theta}|\omega_4\rangle.
    \end{equation}
    
    \subsection{2-RDM Eigenvalues of Special Symmetric States}
    \label{apd:sub:eigen}
    
    We now analyze a more general class of four-qubit symmetric states of the form:
    \begin{equation}
    \begin{aligned}
        \ket{\psi} = c_0 \ket{\omega_0} + c_2 \ket{\omega_2} + c_4 \ket{\omega_4}
        \label{apd:psi_}
    \end{aligned}
    \end{equation}
    where $c_0, c_2, c_4 \in \mathbb{R}$. For such states, the 2-RDMs can be expressed as
    \begin{equation}
    		\rho_{2,\psi} =
    		\begin{pmatrix}
    			a_0^2+\frac{a_2^2}{6} & 0 & 0 & \frac{a_2(a_0+a_4)}{\sqrt{6}} \\
    			0 & a_2^2/3 & a_2^2/3 & 0 \\
    			0 & a_2^2/3 & a_2^2/3 & 0 \\
    			\frac{a_2(a_0+a_4)}{\sqrt{6}} & 0 & 0 & a_4^2+\frac{a_2^2}{6} \\
    		\end{pmatrix}
    \label{apd:psi_2rdm}
    \end{equation}

    \quad
    
    Previous work has proven that a pure symmetric state in the form (\ref{apd:psi_}) is UDP by its 2-RDMs if each eigenvalue of its 2-RDMs is non-degenerate. In other words, if a pure symmetric state in the form (\ref{apd:psi_}) is non-UDP by its 2-RDMs, there must exist at least one pair of degenerate eigenvalues in its 2-RDMs.
    
   To further investigate the uniqueness properties of such states, we examine the degeneracy of the eigenvalues of their 2-RDMs. The eigenvalues of $\rho_{2,\psi}$ can be obtained by solving the characteristic equation:
    \begin{equation}
        \det(RDM_{2,\psi} - \lambda I) = 0.
    \end{equation}
    This equation factors into two separate equations:
    \begin{equation}
    	(\frac{a_2^2}{3}-\lambda)^2-(\frac{a_2^2}{3})^2 = 0,
    	\label{apd:char1}
    \end{equation}
    \begin{equation}
    	(a_0^2+\frac{a_2^2}{6}-\lambda)(a_4^2+\frac{a_2^2}{6}-\lambda)-(\frac{a_2(a_0+a_4)}{\sqrt{6}})^2 = 0.
    	\label{apd:char2}
    \end{equation}
    Solving equation (\ref{apd:char1}), we get two eigenvalues of the 2-RDM:
    \begin{equation}
    	\lambda_1 = 0, \quad
    	\lambda_2 = \frac{2 a_2^2}{3}.
    	\label{apd:lamb1and2}
    \end{equation}
    Next, we rearrange equation (\ref{apd:char2})  into the standard quadratic form of quadratic equation:
    \begin{equation}
    	\lambda^2-(a_0^2+a_4^2+\frac{a_2^2}{3})\lambda-(a_0 a_4 - \frac{a_2^2}{6})^2 = 0,
    \end{equation}
    from which we compute the discriminant to be
    \begin{equation}
    	(a_0+a_4)^2((a_0-a_4)^2+\frac{2a_2^2}{3}) \ge 0.
    \end{equation}
    indicating that the quadratic equation has two real roots. Therefore, the remaining two eigenvalues of the 2-RDM are
    \begin{equation}
        \lambda_{3,4} = \frac{(a_0^2+a_4^2+\frac{a_2^2}{3}) \pm \sqrt{(a_0+a_4)^2((a_0-a_4)^2+\frac{2a_2^2}{3})}}{2}.
    	\label{apd:lamb3and4}
    \end{equation}
    For the 2-RDM in the form (\ref{apd:psi_2rdm}), eigenvalue degeneracies can occur in four distinct  ways:
    \begin{itemize}
        \item (i) $\lambda_1 = \lambda_2$,
        \item (ii) $\lambda_3 = \lambda_4$,
        \item (iii) $\lambda_1 = \lambda_3$ or $\lambda_1 = \lambda_4$,
        \item (iv) $\lambda_2 = \lambda_3$ or $\lambda_2 = \lambda_4$.
    \end{itemize}

    \subsection{Rank Reduction by Target State Properties}
    \label{apd:sub:rank}

    Denote $\mathbf{A}$ as the measurement framework containing all the non-identity single-qubit or two-qubit Pauli operators in four-qubit system, i.e. $|\mathbf{A}|=66$. Based on Corollary~\ref{coro}, for any target state under the measurement framework $\mathbf{A}$, the UDA problem only requires considering density matrices whose ranks are bounded by $8$.

    In this subsection, we demonstrate that for a symmetric pure state $\ket{\psi}$ of the form (\ref{apd:psi_}) as the target state, the maximum rank required to solve the UDA problem in a four-qubit system can be reduced from $8$ to $5$. This reduction is achieved by leveraging symmetry properties and constraints imposed by the measurement outcomes. The reasoning proceeds through the following three theorems.

    \begin{theorem}[Preservation of Containment Relationship under Partial Traces]
    \label{apd:theo_trace}
    
        Let $\rho$ be a density operator on a Hilbert space $\mathcal{H} = \mathcal{H}_P \otimes \mathcal{H}_Q$. If $H$ is a bounded linear operator such that 
        \begin{align*}
            H \in B(\operatorname{supp}\rho),
        \end{align*}
        then its partial trace over $\mathcal{H}_P$ satisfies
        \begin{align*}
            \Tr_P(H) \in B(\operatorname{supp} \Tr_P(\rho)).
        \end{align*}
    \end{theorem}
    
    \begin{proof}
        Given that $\rho$ is a density operator on $\mathcal{H} = \mathcal{H}_P \otimes \mathcal{H}_Q$, consider the subspace $\operatorname{supp} \rho$ (i.e. the support of $\rho$). Let $H$ be an operator in $B(\operatorname{supp} \rho)$. This means that
        \begin{align*}
           H \ket{\psi} = 0 
           \text{\, for all } \ket{\psi} \in (\operatorname{supp} \rho)^\perp.
        \end{align*}

        By the partial trace over $\mathcal{H}_P$, any non-zero vector in $\operatorname{supp} \rho$ must project onto a non-zero vector in $\operatorname{supp} \Tr_P(\rho)$ on subsystem $Q$. Conversely, if a state $\ket{\psi_Q} \in \mathcal{H}_Q$ satisfies
        \begin{align*}
            \ket{\psi_Q} \in (\operatorname{supp} \Tr_P(\rho))^\perp,
        \end{align*}
        then we can have
        \begin{align*}
           \ket{\psi_P} \otimes \ket{\psi_Q} \in (\operatorname{supp} \rho)^\perp 
           \text{\, for all } 
           \ket{\psi_P} \in \mathcal{H}_P.
        \end{align*}
        Since $H \in B(\operatorname{supp} \rho)$, it follows that
        \begin{align*}
            H(\ket{\psi_P} \otimes \ket{\psi_Q}) = 0
            \,\text{ if }\,
            \ket{\psi_Q} \in (\operatorname{supp} \Tr_P(\rho))^\perp.
        \end{align*}
        Recall the definition of partial trace, 
        \begin{align*}
            \Tr_P(H) = \sum\nolimits_{j} \Bigl( \bra{j} \otimes I \Big) \,H\, \Big(\ket{j} \otimes I \Big)
        \end{align*}
        where $
        \{\ket{j}\}$ is an orthonormal basis for $\mathcal{H}_P$. Then, for any $\ket{\psi_Q} \in (\operatorname{supp} \Tr_P(\rho))^\perp$, we calculate
        \begin{align*}
            \Tr_P(H) \ket{\psi_Q} = \sum\nolimits_{j} \Bigl(\bra{j} \otimes I \Big) \,H\, \Big(\ket{j} \otimes I \Big) \ket{\psi_Q}.
        \end{align*}
        Substituting $\Big(\ket{j} \otimes I \Big) \ket{\psi_Q} = \ket{j} \otimes \ket{\psi_Q}$, we find
        \begin{align*}
            \Tr_P(H) \ket{\psi_Q} =\sum\nolimits_{j} \Bigl(\bra{j} \otimes I \Big) \Big(H\, (\ket{j} \otimes \ket{\psi_Q} )\Big)
        \end{align*}
        Since $H\, (\ket{j} \otimes \ket{\psi_Q})=0$ when $\ket{\psi_Q} \in (\operatorname{supp} \Tr_P(\rho))^\perp$,
        \begin{align*}
            \Tr_P(H) \ket{\psi_Q} = 0 
            \,\text{ if }\, 
            \ket{\psi_Q} \in (\operatorname{supp} \Tr_P(\rho))^\perp.
        \end{align*}
        Therefore, we can prove that 
        \begin{align*}
           H \in B(\operatorname{supp} \rho)
           \implies
           \Tr_P(H) \in B(\operatorname{supp} \Tr_P(\rho)).
        \end{align*}

    \end{proof}

    \begin{theorem} [Prior Information from Target Symmetric States]
        \label{apd:theo_prior}
        Let $\rho$ be a four-qubit density matrix that shares identical 2-RDMs with a symmetric pure state $\ket{\psi}$ of the form (\ref{apd:psi_}). Then, for any operator $H \in B(\operatorname{supp} \rho)$, the following conditions hold:
        \begin{itemize}
            \item $\Tr(X_j H) = \Tr(X_k H)$,
            \item $\Tr(Y_j H) = \Tr(Y_k H)$,
            \item $\Tr(Z_j H) = \Tr(Z_k H)$,
            \item $\Tr(Y_j Z_k H) = \Tr(Z_j Y_k H)$
            \item $\Tr(Z_j X_k H) = \Tr(X_j Z_k H)$
            \item $\Tr(X_j Y_k H) = \Tr(Y_j X_k H)$
            \item $\Tr(H) = \Tr(X_j X_k H)+\Tr(Y_j Y_k H)+\Tr(Z_j Z_k H)$.
        \end{itemize}
        for any qubit pair $\{j,k\}$ in four-qubit system.
  
    \end{theorem}

    \begin{proof}
    
    Each 2-RDM of $\rho$ is identical to the 2-RDMs of the symmetric pure state $\ket{\psi}$, denoted as $\rho_{2, \psi}$. By Theorem~\ref{apd:theo_trace}, we know that each 2-RDM of $H \in B(\operatorname{supp} \rho)$ lies in
    $B(\operatorname{supp} \rho_{2,\psi})$. 
    
    With at least one zero eigenvalue, the 2-RDM $\rho_{2,\psi}$ in the form (\ref{apd:psi_2rdm}) has rank at most 3. The support of a generic $\rho_{2,\psi}$ is spanned by states $\ket{00}$, $\ket{11}$ and $(\ket{01}+\ket{10})/\sqrt{2}$. herefore, any operator $H \in B(\operatorname{supp} \rho)$ satisfies
    \begin{align*}
        \Tr_{\{3,4\}}(H)(\ket{01} - \ket{10}) = 0,
    \end{align*}
    where $\Tr_{\{3,4\}}(H)$ is the 2-RDM of $H$ obtained by tracing out the last two qubits.
    
    We analyze the implications of this orthogonality constraint by considering the following cases:
    \begin{itemize}
    \item $\bra{00}\Tr_{\{3,4\}}(H)(\ket{01} - \ket{10})=0$ implies that
    \begin{align*}
        \Tr(X_1 H) + \Tr(X_1 Z_2H) 
        - \Tr(X_2 H) - \Tr(Z_1 X_2 H) = 0,
        \\
        \Tr(Y_1 H) + \Tr(Y_1 Z_2 H) -  \Tr(Y_2 H) - \Tr(Z_1 Y_2 H) = 0.
    \end{align*}
    \item $ \bra{11}\Tr_{\{3,4\}}(H)(\ket{01} - \ket{10})=0$ implies that
    \begin{align*}
        \Tr(X_2 H) + \Tr(Z_1 X_2 H) 
        - \Tr(X_1 H) - \Tr(X_1 Z_2 H) = 0,
        \\
        \Tr(Y_2 H) + \Tr(Z_1 Y_2 H) - \Tr(Y_1 H) - \Tr(Y_1 Z_2 H) = 0.
    \end{align*}
    \item $ (\bra{01}+\bra{10})\Tr_{\{3,4\}}(H) (\ket{01} - \ket{10}) = 0$ implies that
    \begin{align*}
        \Tr(H) - \Tr(Z_1 Z_2 H) 
        - &\Tr(X_1 X_2 H) - \Tr(Y_1 Y_2 H) = 0,
        \\
        \Tr(X_1 Y_2 H) & - \Tr (Y_1 X_2 H) = 0.
    \end{align*}
    \item $ (\bra{01}-\bra{10})\Tr_{\{3,4\}}(H) (\ket{01} - \ket{10}) = 0$ implies that
    \begin{align*}
        \Tr(Z_1 H) - \Tr(Z_2 H) = 0.
    \end{align*}
    \end{itemize}
    By rearranging the equations above, we deduce the following constraints:
    \begin{itemize}
            \item $\Tr(X_1 H) = \Tr(X_2 H)$,
            \item $\Tr(Y_1 H) = \Tr(Y_2 H)$,
            \item $\Tr(Z_1 H) = \Tr(Z_2 H)$,
            \item $\Tr(Y_1 Z_2 H) = \Tr(Z_1 Y_2 H)$
            \item $\Tr(Z_1 X_2 H) = \Tr(X_1 Z_2 H)$
            \item $\Tr(X_1 Y_2 H) = \Tr(Y_1 X_2 H)$
            \item $\Tr(H) = \Tr(X_1 X_2 H)+\Tr(Y_1 Y_2 H)+\Tr(Z_1 Z_2 H)$.
        \end{itemize}
    
    Similarly, we can show that these constraints hold for any other qubit pairs in the four-qubit system.
    \end{proof}

    \begin{theorem}[Rank Reduction by Symmetric State Property]
    \label{apd:theo_rank}
    Let $\mathbf{A}$ be a measurement framework containing all single-qubit and two-qubit Pauli operators in a four-qubit system. Given a pure symmetric state $\ket{\psi}$ in the form (\ref{apd:psi_}), if density matrix $\rho$ shares identical 2-RDMs with symmetric state $\ket{\psi}$, then there exists a density matrix $\sigma$ with rank at most $5$ such that
    \begin{align*}
        \Tr(\ket{\psi}\bra{\psi}\sigma) = \Tr(\ket{\psi}\bra{\psi}\sigma) 
        \text{ and }
        \Vec{\mathcal{M}}_{\mathbf{A}}(\sigma) = \Vec{\mathcal{M}}_{\mathbf{A}}(\rho).
    \end{align*}
    \end{theorem}

    \begin{proof}
    Suppose $\rho$ is density matrix sharing identical 2-RDMs with symmetric state $\ket{\psi}$ in the form (\ref{apd:psi_}). Recall the proof of Theorem~\ref{theo}. To obtain a lower-rank density matrix equivalent to $\rho$, we define an extended measurement framework $\mathbf{B} = \mathbf{A} \cup \{ I, \ket{\psi}\bra{\psi}\}$ where $\mathbf{A}$ contains all non-identity single-qubit or two-qubit Pauli operators. Define the subspace of Hermitian operators:
    \begin{align*}
        M = \{ H \in B(\operatorname{supp}\rho): 
        \Vec{\mathcal{M}}_{\mathbf{B}}(H) = \Vec{0}
        \text{ and }
        H=H^{\dagger}
        \}.
    \end{align*}
    The dimension of this set is at least $r^2-|\mathbf{B}|$, where $r = \operatorname{rank}(\rho)$ and $|\mathbf{B}|$ is the number of observables in $\mathbf{B}$.

    In fact, the new measurement framework does not need to include all operators from $\mathbf{A}$. According to Theorem~\ref{apd:theo_prior},  we leverage the symmetry properties of any operator $H \in B(\operatorname{supp} \rho)$ to identify and eliminate redundant operators: 
    \begin{itemize}
        \item \textbf{Single-qubit Pauli operators:}  Since $\Tr(X_j H) = \Tr(X_k H)$, $\Tr(Y_j H) = \Tr(Y_k H)$ and $\Tr(Z_j H) = \Tr(Z_k H)$ for any qubit pair $\{j,k\}$, it suffices to include only three single-qubit operators $X_1,Y_1, Z_1$.
        \item \textbf{Two-qubit Pauli operators (different Pauli matrices on the pair):} For each qubit pair $\{j,k\}$, since $\Tr(Y_1 Z_2 H) = \Tr(Z_1 Y_2 H)$, $\Tr(Z_1 X_2 H) = \Tr(X_1 Z_2 H)$ and $\Tr(X_1 Y_2 H) = \Tr(Y_1 X_2 H)$, if $Y_1 Z_2, Z_1 X_2, X_1 Y_2$ are already included, including $Z_1 Y_2,X_1 Z_2,Y_1 X_2$ is non-necessary.
        \item \textbf{Two-qubit Pauli operators (same Pauli matrix on the pair):} Suppose identity operator $I$ is already included. For each qubit pair $\{j,k\}$, since $\Tr(H) = \Tr(X_j X_k H)+\Tr(Y_j Y_k H)+\Tr(Z_j Z_k H)$, if $X_j X_k$, $Y_j Y_k$ are already included, including $Z_j Z_k$ is non-necessary.
    \end{itemize}

   There are 6 distinct qubit pairs in a four-qubit system, regardless of the order. For each qubit pair $\{j,k\}$, it suffices to including two-qubit Pauli operators $Y_j Z_k,Z_j X_k,X_j Y_k,X_j X_k,Y_j Y_k$. Altogether, this gives $30$ two-qubit Pauli operators. Adding the identity operator $I$, $\ket{\psi} \bra{\psi}$, and the single-qubit operators $X_1$, $Y_1$, $Z_1$, the total number of observables in the reduced framework $\mathbf{C}$ is $35$. Using the reduced measurement framework $\mathbf{C}$, we redefine the subspace $M$ as
    \begin{align*}
        M = \{ H \in B(\operatorname{supp}\rho): 
        \Vec{\mathcal{M}}_{\mathbf{C}}(H) = \Vec{0}
        \text{ and }
        H=H^{\dagger}
        \}
    \end{align*}
    The dimension of $M$ is at least $r^2-|\mathbf{C}|$.

    Similar to the proof of Theorem~\ref{theo}, by iteratively reducing the rank, we may eventually obtain a equivalent density matrix $\sigma$ with $\operatorname{rank}(\sigma) < \sqrt{|\mathbf{C}|+1}$. Since $|\mathbf{C}|=35$, there exists a density matrix $\sigma$ with rank at most 5 such that
    \begin{align*}
        \Tr(\ket{\psi}\bra{\psi}\sigma) = \Tr(\ket{\psi}\bra{\psi}\rho) 
        \text{ and }
        \Vec{\mathcal{M}}_{\mathbf{A}}(\sigma) = \Vec{\mathcal{M}}_{\mathbf{A}}(\rho).
    \end{align*}
    
    %  If $r^2-|\mathbf{C}| \ge 1$, $M$ is non-empty. Similar to the proof of Theorem~\ref{theo}, for any non-zero $H \in M$, there exist $\epsilon > 0$ such that $\sigma = \rho + \epsilon H$ satisfies $\operatorname{rank}(\sigma) < \operatorname{rank}(\rho)$, $\Tr(\ket{\psi}\bra{\psi}\sigma) = \Tr(\ket{\psi}\bra{\psi}\rho)$ and $\Vec{\mathcal{M}}_{\mathbf{A}}(\sigma) = \Vec{\mathcal{M}}_{\mathbf{A}}(\rho)$. By iteratively reducing the rank, we may eventually obtain a equivalent density matrix $\sigma$ with $\operatorname{rank}(\sigma) <  \sqrt{|\mathbf{C}|+1}$ that is equivalent to $\rho$. Since $|\mathbf{C}|=35$, we may always find a equivalent density matrix $\sigma$ with rank at most $5$, such that
    % \begin{align*}
    %     \Tr(\ket{\psi}\bra{\psi}\sigma) = \Tr(\ket{\psi}\bra{\psi}\rho) 
    %     \text{ and }
    %     \Vec{\mathcal{M}}_{\mathbf{A}}(\sigma) = \Vec{\mathcal{M}}_{\mathbf{A}}(\rho).
    % \end{align*}

    \end{proof}

\end{document}